\documentclass[11pt]{article}

\usepackage[a4paper, margin=1in]{geometry}
\usepackage{amsmath, amssymb, amsthm}
\usepackage{algorithm}
\usepackage{algpseudocode}
\usepackage{graphicx}
\usepackage{booktabs}
\usepackage[font=small, labelfont=bf]{caption}
\usepackage[affil-it]{authblk}
\usepackage[colorlinks=true, citecolor=blue, linkcolor=blue, urlcolor=blue]{hyperref}
\usepackage{xcolor}
\usepackage{enumitem}

\newtheorem{proposition}{Proposition}

\title{Decision Support under Prediction-Induced Censoring}

\author[1]{Yan Chen}
\author[1]{Ruyi Huang}
\author[1]{Cheng Liu\thanks{Corresponding Author. Email: cliu647@cityu.edu.hk}}
\affil[1]{Department of Systems Engineering, City University of Hong Kong, Hong Kong SAR, China}

\date{}

\begin{document}
\maketitle

\begin{abstract}
In many data-driven online decision systems, actions determine not only operational costs but also the data availability for future learning---a phenomenon termed \textbf{Prediction-Induced Censoring (PIC)}. This challenge is particularly acute in large-scale resource allocation for generative AI (GenAI) serving: insufficient capacity triggers shortages but hides the true demand, leaving the system with only a "greater-than" constraint. Standard decision-making approaches that rely on uncensored data suffer from selection bias, often locking the system into a self-reinforcing low-provisioning trap. To break this loop, this paper proposes an adaptive approach named PIC-Reinforcement Learning \textbf{(PIC-RL)}, a closed-loop framework that transforms censoring from a data quality problem into a decision signal. PIC-RL integrates (1) \textbf{Uncertainty-Aware Demand Prediction} to manage the information--cost trade-off, (2) \textbf{Pessimistic Surrogate Inference} to construct decision-aligned conservative feedback from shortage events, and (3) \textbf{Dual-Timescale Adaptation} to stabilize online learning against distribution drift. The analysis provides theoretical guarantees that the feedback design corrects the selection bias inherent in naive learning. Experiments on production Alibaba GenAI traces demonstrate that PIC-RL consistently outperforms state-of-the-art baselines, reducing service degradation by up to 50\% while maintaining cost efficiency.
\end{abstract}

\noindent\textbf{Keywords:} Online Decision Support, Resource Management, Censored Information, Reinforcement Learning, Generative AI

\begin{figure*}[t]
  \centering
  \includegraphics[width=\textwidth]{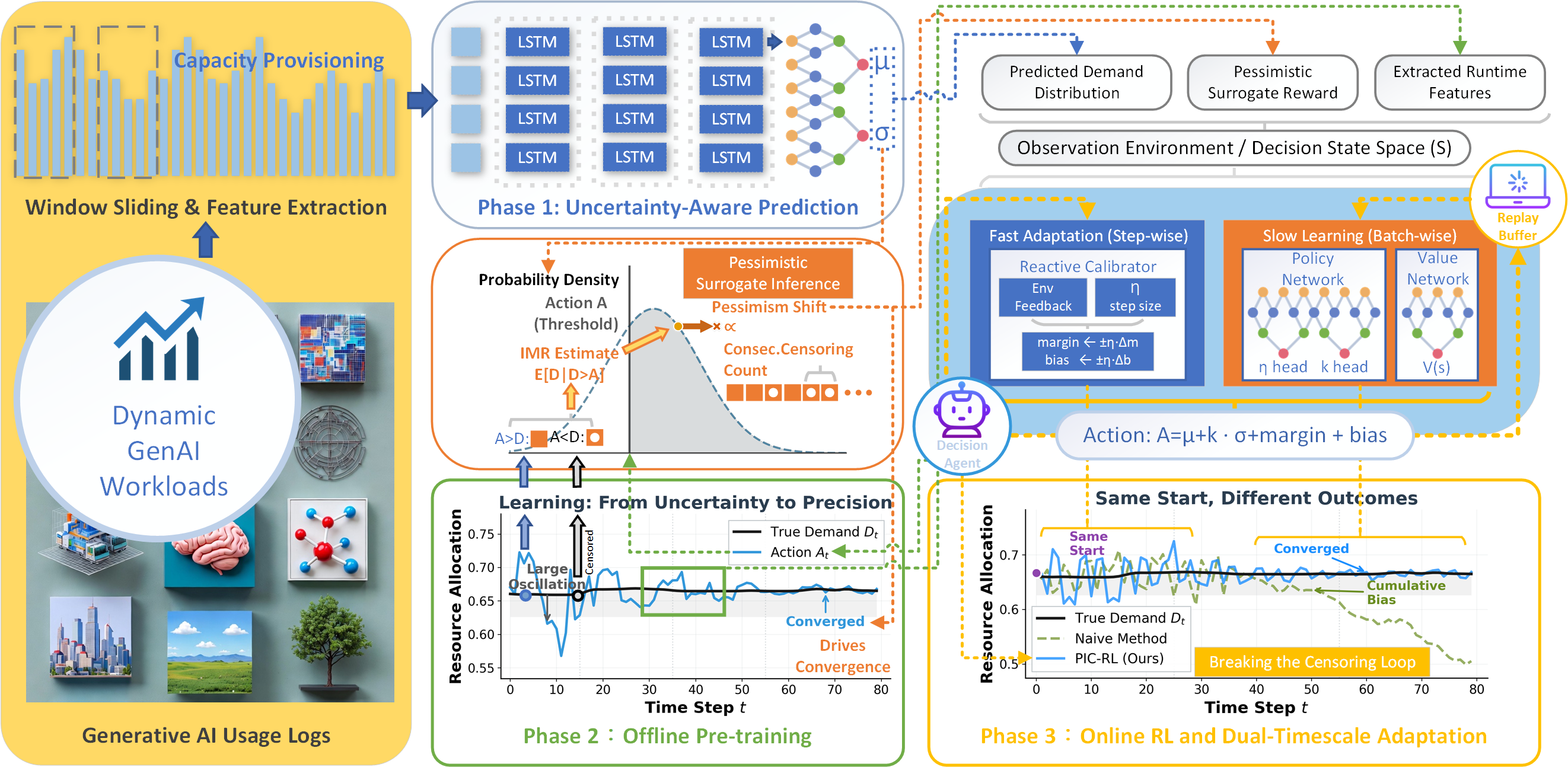}
  \caption{\textbf{The PIC-RL Framework.} A three-phase architecture transforming censoring from a missing-label problem into a supervision signal: (1) Uncertainty-Aware Prediction, (2) Offline Pre-training, and (3) Online RL and Dual-Timescale Online Adaptation.}
  \label{fig:framework}
\end{figure*}

\section{Introduction}
Large-scale cloud services typically rely on the pipeline "online demand forecasting $\rightarrow$ capacity provisioning" to balance utilization and service-level objectives (SLOs). This tension is acute for GenAI and large language model (LLM) inference serving, where providers must meet stringent latency targets while controlling expensive GPU costs. Under volatile workloads, practitioners often over-provision to avoid risks~\cite{ref1}, yet static allocations or black-box scheduling struggle to match heterogeneous resource demands~\cite{ref2,ref3,yuan2023bicritic}. Recent works explore pooling or peak-shaping to improve efficiency~\cite{ref4,ref5}, but they fundamentally rely on feedback loops that break under shortage. In provisioning, actions determine not only cost but also feedback observability: when resources are insufficient ($a_t < d_t$), operators observe only the capacity limit rather than true demand~\cite{ref3}. This motivates a fundamental question: \textbf{how should learning proceed when online actions shape the data being observed?}

This paper studies this problem through the lens of PIC. Unlike standard supervised learning, PIC creates a self-reinforcing loop: underestimation $\rightarrow$ censored feedback $\rightarrow$ biased learning $\rightarrow$ persistent failure~\cite{ref7}. The \textbf{action-as-threshold} mechanism introduces a unique information--cost coupling: acquiring informative labels requires paying higher resource costs~\cite{ref7}. While censored demand is well-studied in inventory management~\cite{ref11,ref12,li2025fliggy} and econometrics~\cite{ref20,ref21}, these fields predominantly assume \textbf{exogenous censoring} driven by external constraints rather than the learner's policy. Consequently, classical base-stock policies~\cite{ref13,ref14} or regression methods (e.g., Tobit~\cite{oneill2024tobit})~\cite{ref22,ref23} fail to model the causal feedback loop where the decision itself induces the data distribution shift. Under nonstationary workloads, the inability to actively manage the information--cost trade-off leads to slow adaptation or collapse~\cite{ref8,ref9,xiang2025parcel}.

To break this loop, \textbf{PIC-RL} is introduced (Figure~\ref{fig:framework}), transforming censoring from a "missing label" problem into a supervision signal. The framework rests on two PIC-aligned innovations. First, \textbf{Pessimistic Surrogate Inference} (Phase 2, Figure~\ref{fig:framework} Middle) synthesizes conservative rewards from censored events to correct the selection bias inherent in uncensored-only learning (Propositions~1--2; see Section~3). Second, \textbf{Dual-Timescale Adaptation} (Phase 3, Figure~\ref{fig:framework} Right) couples a fast \textit{Reactive Calibrator} for immediate volatility with a slow \textit{Policy Network} for robust strategy refinement. Experiments on Alibaba GenAI traces show that PIC-RL prevents the passive service degradation observed in baselines~\cite{ref29}, providing a foundation for learning under endogenous feedback constraints.

\section{Problem Definition}

Consider a discrete-time horizon \(t=1,2,\dots,T\). At each time step, there is a latent normalized demand \(d_t\in[0,1]\) and the system chooses an action \(a_t\in[0,1]\), interpreted as a provisioning level. The action affects cost immediately and, crucially in PIC, determines the feedback observability. Let \(H_t\) denote the history available before acting at time \(t\), and let \(\pi\) be a policy that selects actions as \(a_t=\pi(H_t)\).

The observation model follows the action-as-threshold structure. After choosing \(a_t\), the system observes the truncated value \(y_t=\min(d_t,a_t)\) together with a censoring indicator \(c_t=\mathbf{1}[d_t>a_t]\). When \(c_t=0\), the system observes the ground truth \(y_t=d_t\); when \(c_t=1\), it only observes \(y_t=a_t\) and the inequality \(d_t>a_t\). Thus, the online feedback is \((y_t,c_t)\), meaning the availability of ground-truth labels is causally coupled to the decision itself.

Decision quality is evaluated using an asymmetric linear cost. For any \((a_t,d_t)\), the per-step cost is
\begin{equation}
\ell(a_t,d_t)=c_{\mathrm{under}}(d_t-a_t)_+ + c_{\mathrm{over}}(a_t-d_t)_+,
\end{equation}
where \(c_{\mathrm{under}}>c_{\mathrm{over}}\) (typically \(c_{\mathrm{under}}=2, c_{\mathrm{over}}=1\)). The primary metric is the cumulative regret:
\begin{equation}
\mathrm{Regret}(\pi)\triangleq \sum_{t=1}^{T}\ell(a_t,d_t).
\label{eq:regret}
\end{equation}
The Mean Absolute Error (MAE) is also reported to measure tracking accuracy.

\paragraph{Protocol: Offline-to-Online under Drift.}
Access to historical data is assumed for offline pretraining. However, the online demand process may undergo distribution shift (concept drift). The system must adapt online using only the partial PIC feedback \((y_t,c_t)\), without access to immediate ground truth when censored. This setting captures the core challenge: warm-starting from historical logs while adapting to nonstationary demand under action-dependent censoring.

\section{Methodology}

\textbf{PIC-RL} is a closed-loop framework designed to realign "decision $\rightarrow$ observability $\rightarrow$ learning signal" under PIC. As illustrated in Figure~\ref{fig:framework}, it operates in three progressive phases. Phase 1 (Uncertainty-Aware Prediction) leverages historical logs to train an uncertainty-aware demand predictor. Phase 2 (Offline Pre-training) tackles censoring by constructing a \textit{Pessimistic Surrogate Inference} mechanism, converting censorship constraints into directionally correct supervision signals for policy initialization. Finally, Phase 3 (Online Reinforcement Learning (RL) and Dual-Timescale Adaptation) deploys a dual-timescale mechanism: a \textit{Reactive Calibrator} (fast loop) handles immediate feedback volatility, while a \textit{Slow Policy Network} (outer loop) refines the strategy using replayed experiences.

\subsection{Phase 1: Uncertainty-Aware Prediction}

Standard point predictors are insufficient under PIC because they fail to quantify the risk of censoring. To enable risk-aware provisioning, a probabilistic predictor \(f_\theta\) is trained to output a distribution over future demand, parameterized by a mean \(\mu_t\) and a standard deviation \(\sigma_t\). Implemented as an LSTM~\cite{jia2017incremental} with a Gaussian likelihood head, the model minimizes the negative log-likelihood (NLL) on historical data:
\begin{equation}
\mathcal{L}_{\text{pred}} = \frac{1}{N}\sum_{i=1}^{N} \left[ \log \sigma_i + \frac{1}{2}\left(\frac{d_i - \mu_i}{\sigma_i}\right)^2 \right].
\end{equation}
Crucially, the learned uncertainty \(\sigma_t\) serves as a dynamic signal for the downstream policy: when \(\sigma_t\) is high, the system can choose to pay a higher resource cost (via a safety buffer) to reduce censoring risk and acquire more informative feedback.

\subsection{Phase 2: Offline Pretraining}

Phase 2 pretrains a policy network using offline rollouts on the training data. The key innovation is a theoretically grounded surrogate reward for censored steps, which allows the policy to learn from censoring events rather than ignoring them.

\subsubsection{Policy and Value Networks}

The policy network \(\pi_\phi\) takes a state vector \(s_t\) as input and outputs two quantities:
\begin{equation}
(\eta_t, k_t) = \pi_\phi(s_t),
\end{equation}
where:
\begin{itemize}
  \item \(\eta_t \in [0.5, 3.0]\) is a step-size multiplier that controls the speed of fast calibration updates in Phase 3.
  \item \(k_t \in [0, 2.0]\) is the uncertainty utilization coefficient that determines how much of \(\sigma_t\) to add as a safety buffer.
\end{itemize}

A value network \(V_\psi(s_t)\) provides a baseline for variance reduction in policy gradient updates.

\subsubsection{State Design}

The RL agent operates on a rich state space \(s_t\) constructed from extracted runtime features (Figure~\ref{fig:framework}, top-right), providing the policy with sufficient information to make informed decisions under PIC. It includes three categories of features:
\begin{enumerate}
  \item \textbf{Calibration and feedback statistics:} current margin, bias, recent censoring rate, consecutive censoring count, consecutive over-provision count.
  \item \textbf{Sequence statistics:} mean and standard deviation of recent observations, progress through the episode.
  \item \textbf{Prediction and uncertainty estimates:} \(\mu_t, \sigma_t\) from the predictor, plus outputs from a censored-demand estimator (estimated mean, estimated standard deviation, pessimism factor, uncertainty measure).
\end{enumerate}

This design ensures that the policy can "see" its own uncertainty and the current feedback regime, enabling it to adapt its behavior (e.g., become more conservative when information is scarce).

\subsubsection{Learnable Feedback Under Censoring}

A central challenge in Phase 2 is how to assign rewards when censoring occurs and the true demand \(d_t\) is not observed. This is addressed through a theoretically motivated surrogate reward.

\paragraph{The Bias of Uncensored-Only Learning.}
A natural but flawed approach is to simply skip censored steps or impute the censored demand. The following result formalizes why this is problematic.

\begin{proposition}[Instability under Mixture]
Consider a naive learner that fits demand using a mixture of (i) historical uncensored observations and (ii) current censored observations under threshold action $A_t$, with mixture weight $\rho \in (0,1]$ on the censored stream. Its effective target mean is
\(\mu_{\text{mix}}(A_t) = (1-\rho)\mathbb{E}[D] + \rho\, \mathbb{E}[D \mid D \le A_t]\).
If censoring is active ($P(D > A_t) > 0$), then $\mu_{\text{mix}}(A_t) < \mathbb{E}[D]$. Consequently, any update rule that tracks this target (e.g., setting the next base level to $\mu_{\text{mix}}$) exhibits a strictly negative drift even at the unbiased point $A_t=\mathbb{E}[D]$.
In particular, this covers stochastic approximation updates of the form $A_{t+1}=A_t+\gamma_t(\mu_{\text{mix}}(A_t)-A_t)$ with $\gamma_t>0$.
\end{proposition}

\begin{figure}[t]
\centering
\includegraphics[width=\linewidth]{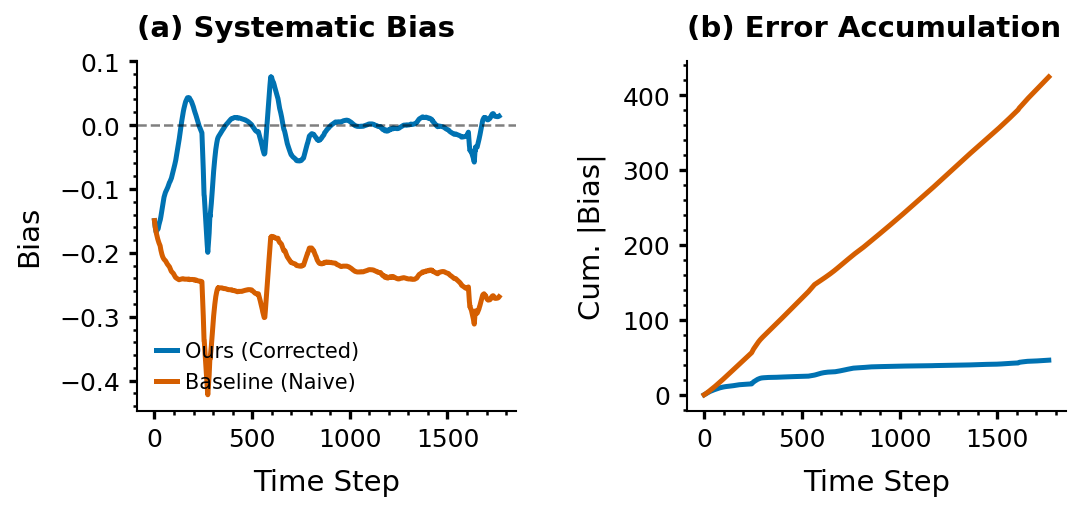}
\caption{Verification of Proposition 1 (Instability). (a) Naive learning exhibits systematic negative bias. (b) Cumulative error confirms that the system inevitably drifts into a ``censoring trap,'' even with historical data replay.}
\label{fig:theory-bias}
\end{figure}

\begin{proof}
Let the true demand expectation be $\mu^* = \mathbb{E}[D]$. The learner's target minimizes risk over the mixed distribution:
\begin{equation}
\mathbb{E}_{\text{train}}[D] = (1-\rho)\mu^* + \rho\, \mathbb{E}[D \mid D \le A_t].
\end{equation}
By the law of total expectation, $\mu^* = \mathbb{E}[D \mid D \le A_t]P(D \le A_t) + \mathbb{E}[D \mid D > A_t]P(D > A_t)$. Since $D > A_t$ implies values strictly greater than those in the observed set, it holds that $\mathbb{E}[D \mid D > A_t] > \mathbb{E}[D \mid D \le A_t]$, which implies $\mu^* > \mathbb{E}[D \mid D \le A_t]$. Let the bias gap be $\Delta(A_t) = \mu^* - \mathbb{E}[D \mid D \le A_t] > 0$.
Substituting back:
\begin{equation}
\mathbb{E}_{\text{train}}[D] = (1-\rho)\mu^* + \rho\, (\mu^* - \Delta(A_t)) = \mu^* - \rho\, \Delta(A_t).
\end{equation}
Thus, $\mathbb{E}_{\text{train}}[D] < \mu^*$. If $A_t = \mu^*$, the new target is strictly lower, creating a negative drift force $-\rho\, \Delta(A_t)$.
\end{proof}

\begin{proposition}[Consistency and Escape of Surrogate Reward]
Let the surrogate reward be \(r(a) \propto -(\hat{\mu} + \hat{\sigma}\lambda(z) - a) \cdot \Psi(n)\), where \(z \triangleq (a-\hat{\mu})/\hat{\sigma}\). Under the Gaussian assumption:
\begin{enumerate}
    \item \textbf{Gradient Consistency:} \(\frac{\partial r}{\partial a} > 0\) for all \(a\). The gradient consistently incentivizes increasing the action when censored.
    \item \textbf{Pessimistic Escape:} For a fixed action \(a\), as the consecutive censoring count \(n\) increases, the gradient magnitude \(|\frac{\partial r}{\partial a}|\) scales with \(\Psi(n)\), guaranteeing a growing force to escape persistent under-provisioning traps.
\end{enumerate}
\end{proposition}

\begin{figure}[t]
\centering
\includegraphics[width=\linewidth]{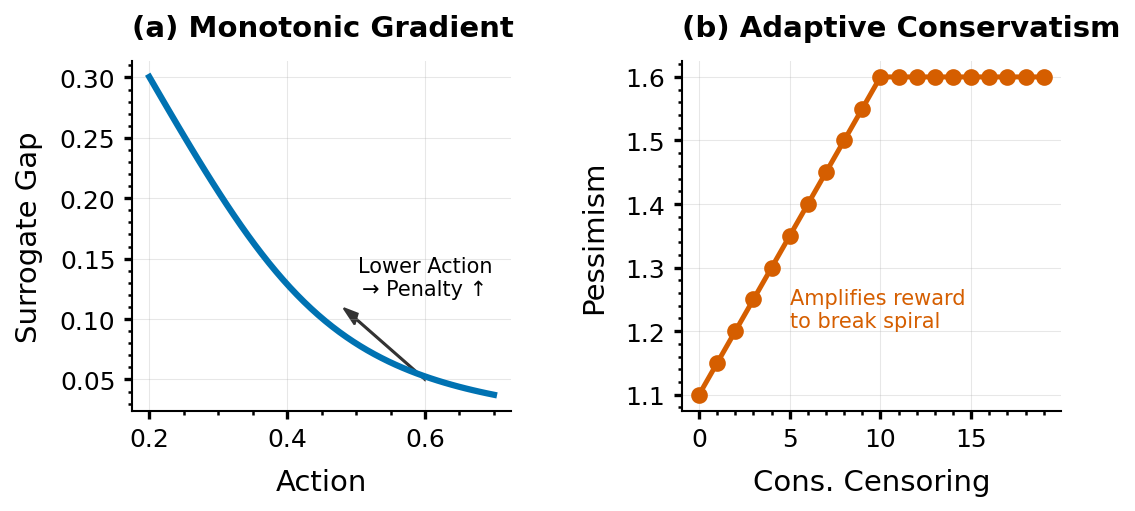}
\caption{Mechanisms of Proposition 2. (a) Strict monotonicity of the surrogate gap ensures gradient consistency ($\partial r/\partial a > 0$). (b) The pessimism factor $\Psi(n)$ amplifies reward signals super-linearly to enable escape from censoring traps.}
\label{fig:theory-mechanism}
\end{figure}

\begin{proof}
Let \(G(a) = \hat{\mu} + \hat{\sigma}\lambda(\frac{a-\hat{\mu}}{\hat{\sigma}}) - a\). The reward is \(r(a) = -C \cdot G(a) \cdot \Psi\), where \(C > 0\) is a constant. Differentiating \(G(a)\) with respect to \(a\):
\begin{equation}
\frac{\partial G}{\partial a} = \hat{\sigma} \lambda'(z) \frac{1}{\hat{\sigma}} - 1 = \lambda'(z) - 1.
\end{equation}
A key property of the Inverse Mills Ratio is that \(0 < \lambda'(z) < 1\) for all \(z \in \mathbb{R}\)~\cite{baricz2008mills}. Therefore, \(\frac{\partial G}{\partial a} \in (-1, 0)\), meaning the estimated gap strictly decreases as action increases. Consequently, the reward gradient is:
\begin{equation}
\frac{\partial r}{\partial a} = -C \Psi (\lambda'(z) - 1) = C \Psi (1 - \lambda'(z)) > 0.
\end{equation}
This confirms \textbf{Consistency}: the policy always receives a positive reward signal for raising actions. Furthermore, since \(1 - \lambda'(z)\) is bounded away from zero locally, the gradient magnitude is proportional to \(\Psi(n)\). As \(n \to N_{\max}\), \(\Psi(n)\) grows, providing the \textbf{Escape} property.
\end{proof}

\paragraph{Surrogate Reward Construction (Pessimistic Surrogate Inference).}
For censored steps, the \textit{Pessimistic Surrogate Inference} mechanism (Figure~\ref{fig:framework}, Middle) uses the Inverse Mills Ratio (IMR)~\cite{heien1990demand} to define the surrogate reward as:
\begin{equation}
r_t^{\text{cens}} = -c_{\text{under}} \cdot \underbrace{\left( \hat{\mu}_t + \hat{\sigma}_t \lambda(\hat{\alpha}_t) - a_t \right)}_{\text{Expected Gap (IMR)}} \cdot \underbrace{\Psi(n_t)}_{\text{Pessimism}},
\end{equation}
where \(\Psi(n_t) = 1 + \beta \cdot \min(n_t, N_{\max})\) scales non-linearly with the duration of information blackout.
For uncensored steps, the true asymmetric cost is used:
\begin{equation}
r_t^{\text{uncens}} = -\ell(a_t, d_t) = -c_{\text{under}}(d_t - a_t)_+ - c_{\text{over}}(a_t - d_t)_+.
\end{equation}
This surrogate reward allows the policy to receive informative gradients even when ground truth is hidden, addressing the core PIC challenge.

Here \(\Psi(n_t)\) represents the pessimism factor that creates a ``break-out'' gradient: if the system gets stuck in a censoring loop, the surging pessimism forces the policy to drastically raise actions, re-acquiring ground truth.

\subsubsection{Offline RL Training}

Rollouts are performed on the training data to simulate the PIC observation mechanism: at each step, censoring is recorded and the observation sequence is updated accordingly (appending \(d_t\) if uncensored, \(a_t\) if censored). The policy and value networks are updated using an Actor--Critic algorithm~\cite{grondman2012survey} with the surrogate rewards defined above.

The key benefit of Phase 2 is that the policy learns ``how to behave under censoring'' before facing real PIC feedback. This warm-start is critical: without it, Phase 3 would begin with a randomly initialized policy that may get trapped in the self-reinforcing underestimation loop.

\subsection{Phase 3: Online RL and Dual-Timescale Adaptation}

In online deployment, the system must adapt to distribution drift while managing the immediate feedback loop. The provisioning action \(a_t\) is formulated as a hierarchical composition of three components:
\begin{equation}
a_t = \underbrace{\mu_t}_{\text{Base}} + \underbrace{k_t(s_t) \cdot \sigma_t}_{\text{Risk Buffer}} + \underbrace{\Delta_t}_{\text{Reactive Correction}}.
\end{equation}
Here, \((\mu_t, \sigma_t)\) are from the frozen predictor. The \textbf{Risk Buffer} is controlled by the RL policy output \(k_t\), allowing the agent to dynamically trade cost for information based on state uncertainty. The \textbf{Reactive Correction} \(\Delta_t = m_t + b_t\) is updated by the \textit{Reactive Calibrator} (Figure~\ref{fig:framework} Right) at every step:
\begin{equation}
\begin{aligned}
m_{t+1} &= m_t + \eta_t(s_t) \delta_m \cdot [\mathbb{I}(c_t) - \mathbb{I}(\neg c_t \land a_t > y_t)], \\
b_{t+1} &= b_t + \eta_t(s_t) \delta_b \cdot [\mathbb{I}(c_t) - \gamma \mathbb{I}(\neg c_t \land a_t > y_t)].
\end{aligned}
\label{eq:calibrator-update}
\end{equation}
This dual-timescale design ensures stability: the fast inner loop (\(\Delta_t\)) handles high-frequency volatility, while the slow outer loop (Policy \(\pi_\phi\)) optimizes structural parameters \((\eta_t, k_t)\).

\emph{Remark (Dual-Timescale Stability).} The interaction between the fast calibrator (\(\Delta_t\)) and slow policy (\(\pi_\phi\)) forms a two-timescale stochastic approximation system. Under standard Lipschitz assumptions, the separation of time scales ensures asymptotic stability~\cite{borkar_stochastic} (Theorem~\ref{thm:stability}).

\paragraph{Theorem 1 (Dual-Timescale Stability).}
\label{thm:stability}
For analysis, consider diminishing step sizes (standard in stochastic approximation); in implementation, small constants are used to track drift. Let $x_t$ denote slow policy parameters and $y_t \triangleq (m_t,b_t)$ denote fast calibration variables (Eq.~\ref{eq:calibrator-update}). Consider the coupled updates
\begin{align}
    x_{t+1} &= x_t + \alpha_t [H(x_t, y_t) + M_{t+1}], \\
    y_{t+1} &= y_t + \beta_t [G(x_t, y_t) + N_{t+1}],
\end{align}
where $\{M_t\}, \{N_t\}$ are martingale difference sequences.

\emph{Assumptions.} (A1) The step sizes satisfy Robbins--Monro conditions:
\begin{equation}
\sum_t \alpha_t = \sum_t \beta_t = \infty,\qquad \sum_t (\alpha_t^2 + \beta_t^2) < \infty,\qquad \alpha_t/\beta_t \to 0.
\end{equation}
(A2) The functions $H, G$ are bounded and locally Lipschitz almost everywhere; the fast loop admits an ordinary differential equation (ODE) (or differential inclusion) representation.
(A3) For any fixed $x$, the fast dynamics ODE $\dot{y}(\tau) = G(x, y(\tau))$ has a unique globally asymptotically stable equilibrium $\lambda(x)$.

\emph{Proof sketch.} Consider the Lyapunov function $V_x(y) = \frac{1}{2}\|y - \lambda(x)\|^2$. Under (A2)--(A3), $V_x$ strictly decreases along trajectories of the fast ODE whenever $y \neq \lambda(x)$, hence $y_t$ tracks $\lambda(x_t)$ on the fast timescale~\cite{borkar_stochastic,kushner2003stochastic}. Given this tracking, the slow dynamics asymptotically follow the reduced ODE $\dot{x}(t) = H(x(t), \lambda(x(t)))$, and the two-timescale theorem implies that the coupled iterates $(x_t,y_t)$ converge (or track, under nonstationarity) the stable invariant set of this reduced system~\cite{borkar_stochastic,kushner2003stochastic}.

\begin{algorithm}[t]
\caption{PIC-RL: Learning with Prediction-Induced Censoring}
\label{alg:pic_rl}
\begin{algorithmic}[1]
\Require Historical data $\mathcal{D}_{\text{train}}$, Online horizon $T$, Hyperparameters $\delta_m, \delta_b, \gamma, N_{\text{update}}$
\Ensure Online provisioning decisions $a_{1:T}$

\Statex \textbf{Phase 1: Uncertainty-Aware Prediction (Offline)}
\State Train predictor $f_\theta$ via NLL minimization:
\State $\quad \theta^* \leftarrow \arg\min_\theta \sum_{(x_i, d_i) \in \mathcal{D}_{\text{train}}} \mathcal{L}_{\text{NLL}}\left(d_i; f_\theta(x_i)\right)$

\Statex \textbf{Phase 2: Offline Pre-training with Pessimistic Surrogate Inference (Offline)}
\State Initialize policy $\pi_\phi$, value network $V_\psi$
\For{epoch $= 1, \dots, E$}
    \State Rollout trajectory using simulator with action $a_\tau \sim \pi_\phi(s_\tau)$
    \For{each step $\tau$ in rollout}
        \If{censored ($c_\tau=1$)}
            \State Construct \textbf{Pessimistic Surrogate Reward}:
            \State $r_\tau \leftarrow -c_{\text{under}} \cdot \widehat{\text{gap}}(a_\tau; \hat{\mu}_\tau, \hat{\sigma}_\tau) \cdot \text{pess}(n_{\text{cens}})$ \Comment{Prop. 2}
        \Else
            \State Use true cost: $r_\tau \leftarrow -\ell(a_\tau, d_\tau)$
        \EndIf
    \EndFor
    \State Update $\phi, \psi$ via Actor-Critic on $\{r_\tau\}$ sequences
\EndFor

\Statex \textbf{Phase 3: Online RL and Dual-Timescale Adaptation (Online)}
\State Initialize fast calibration terms: $m_0 \leftarrow 0, b_0 \leftarrow 0$
\State Initialize Replay Buffer $\mathcal{B}$
\For{time step $t = 1, \dots, T$}
    \State \textbf{Predict:} $(\mu_t, \sigma_t) \leftarrow f_{\theta^*}(\text{window}_t)$
    \State \textbf{Plan:} $(\eta_t, k_t) \leftarrow \pi_\phi(s_t)$ \Comment{Slow Variable: Policy Output}
    \State \textbf{Act:} $a_t \leftarrow \text{clip}(\mu_t + k_t \sigma_t + m_t + b_t, 0, 1)$
    \State \textbf{Observe:} Censored feedback $(y_t, c_t)$ under the action-as-threshold model
    
    \State \textbf{Fast Update (Event-Driven Calibration):}
    \If{$c_t = 1$} \Comment{Under-provisioning (Shortage)}
        \State $m_{t+1} \leftarrow m_t + \eta_t \delta_m, \quad b_{t+1} \leftarrow b_t + \eta_t \delta_b$
    \Else \Comment{Over-provisioning}
        \State $m_{t+1} \leftarrow m_t - \eta_t \delta_m, \quad b_{t+1} \leftarrow b_t - \gamma \eta_t \delta_b$
    \EndIf
    
    \State \textbf{Reward \& Store:}
    \If{$c_t = 1$}
        \State $r_t \leftarrow -c_{\text{under}} \cdot \widehat{\text{gap}}(a_t; \hat{\mu}_t, \hat{\sigma}_t) \cdot \text{pess}(n_{\text{cens}})$ \Comment{Prop. 2}
    \Else
        \State $r_t \leftarrow -\ell(a_t, y_t)$
    \EndIf
    \State Append $(s_t, a_t, r_t, s_{t+1})$ to $\mathcal{B}$
    
    \State \textbf{Slow Update (Policy Refinement):}
    \If{$t \mod N_{\text{update}} == 0$}
        \State Update $\pi_\phi$ using $\mathcal{B}$ with KL regularization~\cite{vieillard2020leverage}
    \EndIf
    \State Update observation window with $(y_t, c_t)$
\EndFor
\end{algorithmic}
\end{algorithm}

\section{Experiments}
\label{sec:experiments}

PIC-RL is evaluated on three production traces against state-of-the-art forecasting, RL, and decision-focused learning baselines. Results show superior tracking accuracy and cost efficiency by learning from censored feedback.

\subsection{Experimental Setup}
\label{sec:setup}

\subsubsection{Production Workloads}
The evaluation uses three distinct industrial traces from Alibaba Cloud~\cite{yang2025prism,yan2026FlexPipe}, representing a spectrum from predictable patterns to highly stochastic dynamics. Table~\ref{tab:datasets} summarizes their statistical characteristics.

\begin{itemize}
    \item \textbf{Dataset A (Industrial DLRM):} GPU request logs from large-scale recommendation services~\cite{yang2025prism}. This workload exhibits strong diurnal seasonality (autocorrelation, Autocorr \(\approx 1.0\)) and moderate variability, serving as a baseline for predictable demand.
    \item \textbf{Dataset B (GenAI-Mem):} GPU memory usage from production Stable Diffusion serving clusters~\cite{yan2026FlexPipe}. Unlike DLRM, this workload is characterized by \textbf{extreme non-stationarity} (coefficient of variation, CV=0.87) and heavy-tailed bursts driven by batched inference scheduling, posing a severe challenge for censorship handling.
    \item \textbf{Dataset C (Context-Rich GenAI):} An augmented version of Dataset B that includes GPU duty-cycle metrics as auxiliary context features~\cite{yan2026FlexPipe}, testing the algorithm's ability to leverage multi-modal signals under partial observability.
\end{itemize}

The resource domain is mapped to \(A_t \in [0,1]\) via min-max normalization using statistics computed \emph{solely} on the training partition; the same scaling is applied to validation and test to prevent look-ahead bias. Values exceeding the training range are clipped to \([0,1]\), simulating continuous resource partitioning mechanisms (e.g., MPS or MIG). Chronological splitting (60/20/20) is used for training, validation, and testing to strictly prevent temporal leakage.

\begin{table}[t]
\centering
\caption{Characteristics of Production Workloads. High peak-to-mean ratio (PMR) and CV highlight the extreme stochasticity of GenAI traces compared to traditional DLRM.}
\label{tab:datasets}
\resizebox{\columnwidth}{!}{
\begin{tabular}{lcccc}
\toprule
\textbf{Dataset} & \textbf{Length} & \textbf{PMR} & \textbf{CV} & \textbf{Workload Dynamics} \\
\midrule
\textbf{DLRM} & 31 days & 2.99 & 0.69 & Strong Seasonality, Predictable \\
\textbf{GenAI-Mem} & 23 hrs & 5.20 & 0.87 & Heavy-tailed, Quasi-periodic Spikes \\
\textbf{GenAI-Ctx} & 23 hrs & 5.20 & 0.87 & Multi-variate, Context-Rich \\
\bottomrule
\end{tabular}
}
\end{table}

\subsubsection{Protocol: Endogenous Censoring Simulation}
To systematically benchmark learning efficacy under endogenous censoring, a \textbf{Counterfactual Simulation Environment} is used. Unlike static supervised evaluation, this environment enforces a dynamic feedback loop: the agent's chosen provision level \(A_t\) instantly determines the data availability for future training steps.
Formally, let \(D_t^*\) denote the unobserved oracle demand (from the trace). At each step \(t\), the environment reveals only the censored feedback:
\begin{equation}
Y_t = \min(D_t^*, A_t), \quad C_t = \mathbb{I}(D_t^* > A_t).
\end{equation}
The agent must update its policy solely based on \((Y_t, C_t)\); the oracle demand \(D_t^*\) is never revealed to the learner and is used only for offline performance evaluation. This protocol rigorously tests the algorithm's ability to break the self-reinforcing bias loop---a capability that cannot be assessed via standard metrics on fixed test sets.

\subsubsection{Baselines}
PIC-RL is compared against state-of-the-art methods via a \textbf{Component-wise Substitution} protocol whenever a baseline targets a specific module. The key challenge is that PIC tightly couples forecasting, control, and observability: many baselines are designed for a \emph{subsystem} (e.g., forecasting or offline RL) and do not specify an end-to-end procedure under endogenous censoring. Running them ``as-is'' would therefore confound missing components with algorithmic weakness. To isolate each baseline's \emph{core mechanism} under the same PIC loop, only the intended phase(s) are substituted while keeping the remaining phases and the counterfactual simulator fixed. For methods that define an end-to-end controller, evaluation is performed as a full-pipeline baseline under the same simulator and cost weights.
\begin{itemize}
    \item \textbf{Phase 1 (Predictor) Substitution:} Replace the Phase-1 probabilistic predictor with \textbf{Temporal Fusion Transformer (TFT)}~\cite{lim2021temporal}, \textbf{Informer}~\cite{zhou2021informer}, or \textbf{Autoformer}~\cite{wu2021autoformer} (with output heads adapted to Gaussian \((\mu_t,\sigma_t)\) for interface compatibility), while keeping Phase 2/3 unchanged.
    \item \textbf{Phase 2 (Offline RL) Substitution:} Replace Phase-2 actor--critic pretraining with \textbf{Conservative Q-Learning (CQL)}~\cite{kumar2020conservative}, while keeping Phase 1 and Phase 3 unchanged.
    \item \textbf{Phase 3 (Online Calibration) Substitution:} Replace Phase-3 margin/bias adaptation with \textbf{Conformal Prediction}~\cite{angelopoulos2023conformal,conformal_newsvendor2024} built on the same Phase-1 predictor; Phase 2 is omitted as Conformal does not define an offline RL pretraining component. We use one-sided (upper) calibration to match the newsvendor-style asymmetric decision objective.
    \item \textbf{Phase 2+3 (Decision Module) Substitution:} Replace the Phase-2/3 decision module with (i) \textbf{Thompson Sampling (TS)}~\cite{agrawal2013thompson} as a Bayesian exploration--exploitation controller, or (ii) a \textbf{Pensieve}-style actor--critic controller~\cite{mao2017neural} by substituting the policy/value networks (including a 1D-CNN state encoder) in both Phase 2 and Phase 3. For fairness, both controllers are adapted only at the interface level to our continuous decision parameterization and state features, while preserving their core exploration/actor--critic mechanisms.
    \item \textbf{Full-Pipeline Baselines:} \textbf{SPO+} (decision-focused training)~\cite{elmachtoub2022smart} and \textbf{Autopilot} (rule-based scaling)~\cite{rzadca2020autopilot}.
\end{itemize}

\subsubsection{Metrics}
\textbf{MAE} is reported as the primary metric for provisioning accuracy: \(\text{MAE} = \frac{1}{T}\sum |A_t - D_t^*|\). Secondary metrics include \textbf{Regret} (Eq.~\ref{eq:regret}), which captures the asymmetric cost of over- and under-provisioning. Lower values indicate better decision quality for both MAE and Regret. Note that Regret depends on specific cost weights (\(c_{\text{under}}=2, c_{\text{over}}=1\)), while MAE provides a generic, transferable measure of tracking quality.

\subsection{Internal Validation and Dynamics}
\label{sec:internal-validation}

Learning dynamics across the three training phases are analyzed on the GenAI-Mem dataset to validate the internal mechanisms of PIC-RL.

\paragraph{Phase 1: Calibrated Uncertainty Quantification.}
The foundation of risk-aware decision-making is a well-calibrated probabilistic predictor. Figure~\ref{fig:phase1-training} shows rapid NLL stabilization and a validation MAE of about 0.0752, indicating reliable estimation of both the conditional mean and dispersion. This yields a point forecast $\mu_t$ together with an uncertainty signal $\sigma_t$, which operationalizes the information--cost trade-off by controlling how aggressively the system buys observability under PIC.

\begin{figure}[t]
\centering
\includegraphics[width=\linewidth]{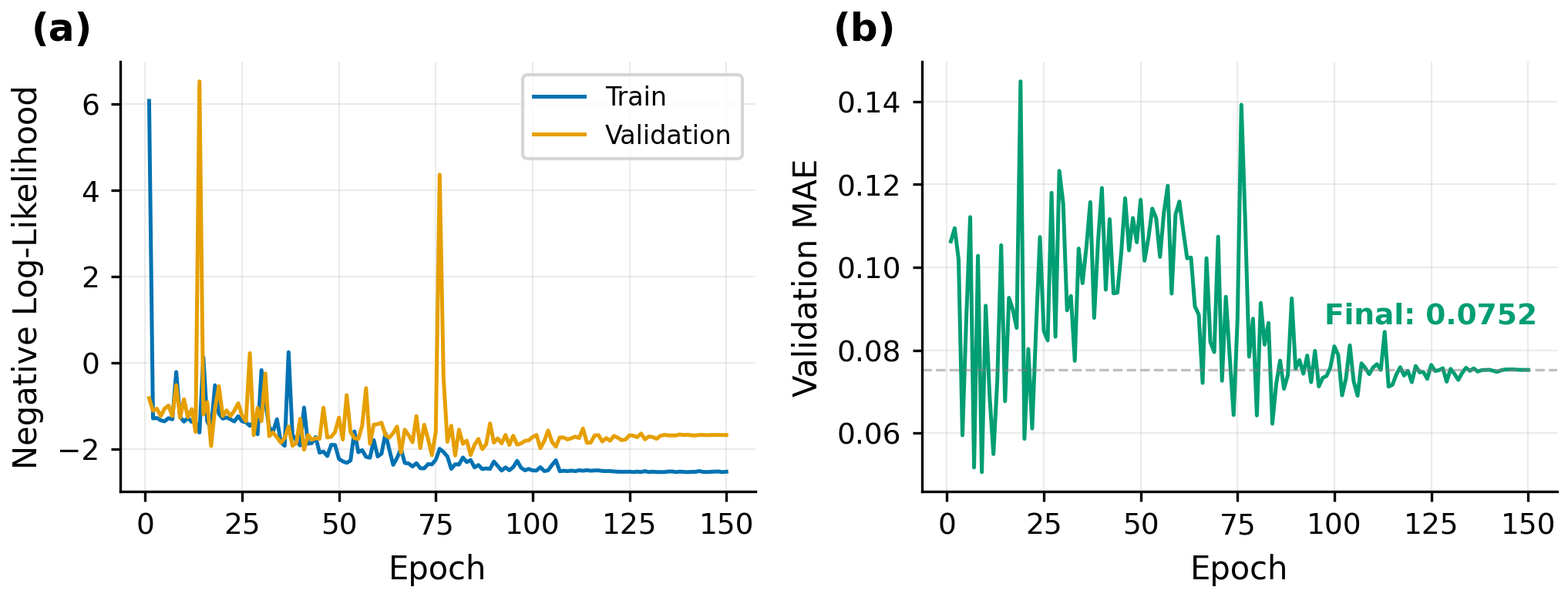}
\caption{Phase 1 Training Dynamics. (a) Rapid NLL convergence confirms the model learns the full demand distribution. (b) Stable validation MAE demonstrates robust generalization without overfitting.}
\label{fig:phase1-training}
\end{figure}

\paragraph{Phase 2: Pessimistic Policy Crystallization.}
In Phase 2, the policy is pretrained with the surrogate reward to warm-start decision-making under censoring. In Figure~\ref{fig:phase2-training}, the critic value loss drops by 79\% over 150 iterations, consistent with the directionally correct learning signal induced by the IMR-based surrogate (Proposition~2). The resulting policy tracks bursty spikes while maintaining a modest buffer, reducing shortages without resorting to persistent over-provisioning.

\begin{figure}[t]
\centering
\includegraphics[width=\linewidth]{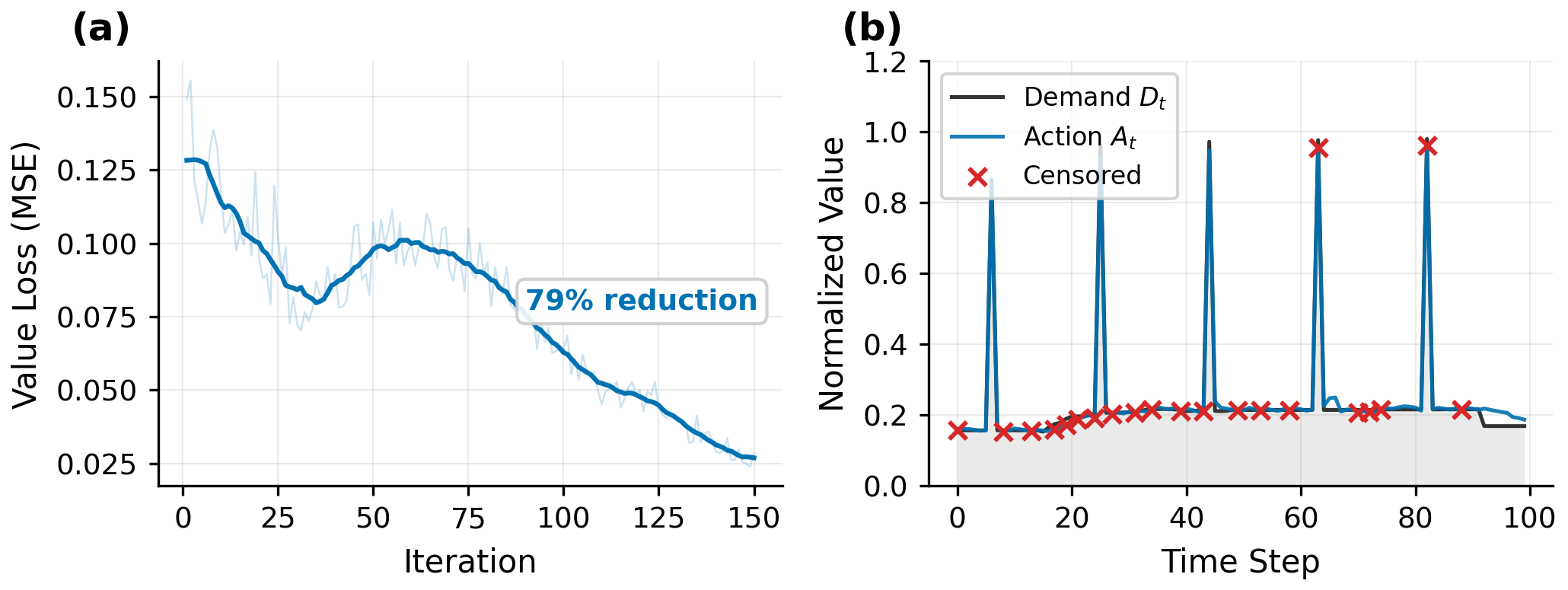}
\caption{Phase 2 Offline Pretraining. (a) 79\% reduction in value loss validates the critic's ability to learn from censored feedback. (b) The derived policy effectively anticipates demand spikes (red crosses indicate censored events).}
\label{fig:phase2-training}
\end{figure}

\paragraph{Phase 3: Dual-Timescale Adaptation.}
Phase 3 evaluates online adaptation under endogenous censoring. Figure~\ref{fig:phase3-dynamics} shows sub-linear growth of cumulative regret under censored feedback, indicating sustained improvement in decision quality. The dual-timescale dynamics separate fast responsiveness (margin) from slow drift-tracking (bias), stabilizing the closed loop while keeping the error distribution $A_t-D_t$ balanced rather than persistently biased toward over-provisioning.

\begin{figure}[t]
\centering
\includegraphics[width=\linewidth]{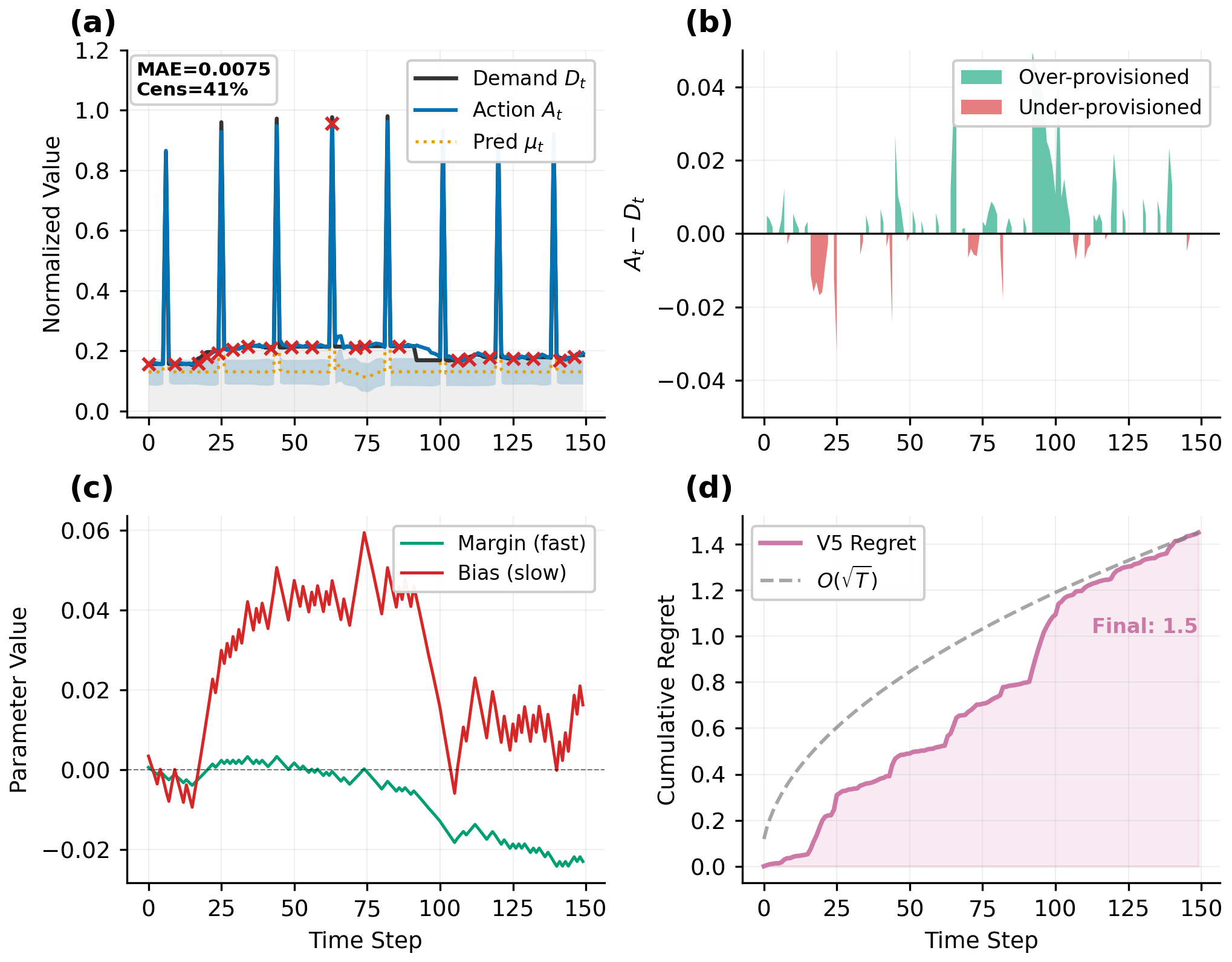}
\caption{Phase 3 Online Performance. (a) Deployment trajectory showing effective tracking within $\mu \pm \sigma$ bounds. (b) Balanced error distribution. (c) Distinct evolution of Fast (margin) and Slow (bias) variables stabilizing the feedback loop. (d) Sub-linear regret growth ($O(\sqrt{T})$) confirming successful adaptation.}
\label{fig:phase3-dynamics}
\end{figure}

\subsection{Main Comparison Results}
\label{sec:comparisons}

Table~\ref{tab:main-results} compares PIC-RL against baselines substituted at different phases under the same PIC loop. The results show that each phase contributes to robust performance under endogenous censoring.

\begin{table}[t]
\centering
\caption{Main Comparison Results. PIC-RL consistently outperforms baselines on MAE and Regret. By systematically substituting key modules (Phase 1, Phase 2, Phase 2+3, or Phase 3), we observe that no single component is sufficient on its own.}
\label{tab:main-results}
\small
\setlength{\tabcolsep}{4pt}
\begin{tabular}{l cc cc cc}
\toprule
& \multicolumn{2}{c}{\textbf{DLRM-GPU}} & \multicolumn{2}{c}{\textbf{GenAI-Ctx}} & \multicolumn{2}{c}{\textbf{GenAI-Mem}} \\
\cmidrule(lr){2-3} \cmidrule(lr){4-5} \cmidrule(lr){6-7}
\textbf{Method} & \textbf{MAE} & \textbf{Reg} & \textbf{MAE} & \textbf{Reg} & \textbf{MAE} & \textbf{Reg} \\
\midrule
\textbf{PIC-RL (Ours)} & \textbf{0.0058} & \textbf{22.8} & \textbf{0.0112} & \textbf{5.4} & \textbf{0.0104} & \textbf{2.4} \\
\midrule
\multicolumn{7}{l}{\textbf{Phase 1 Substitution (Predictor)}} \\
TFT~\cite{lim2021temporal} & 0.0085 & 30.3 & 0.0238 & 6.2 & 0.0113 & 6.6 \\
Informer~\cite{zhou2021informer} & 0.0129 & 37.5 & 0.0195 & 4.3 & 0.0124 & 2.5 \\
Autoformer~\cite{wu2021autoformer} & 0.0112 & 29.0 & 0.1285 & 29.9 & 0.2567 & 54.6 \\
\midrule
\multicolumn{7}{l}{\textbf{Phase 2 Substitution (Offline RL)}} \\
CQL~\cite{kumar2020conservative} & 0.0098 & 27.5 & 0.0385 & 8.3 & 0.0162 & 5.5 \\
\midrule
\multicolumn{7}{l}{\textbf{Phase 3 Substitution (Online Calibration; no Phase-2 pretraining)}} \\
Conformal~\cite{conformal_newsvendor2024} & 0.0072 & 34.4 & 0.0456 & 38.1 & 0.0637 & 33.1 \\
\midrule
\multicolumn{7}{l}{\textbf{Phase 2+3 Substitution (Decision Module)}} \\
Pensieve~\cite{mao2017neural} & 0.0104 & 25.3 & 0.0184 & 6.0 & 0.0154 & 5.0 \\
TS~\cite{agrawal2013thompson} & 0.0112 & 31.8 & 0.2070 & 11.7 & 0.0261 & 15.0 \\
\midrule
\multicolumn{7}{l}{\textbf{Full-Pipeline Baselines}} \\
SPO+~\cite{elmachtoub2022smart} & 0.0158 & 54.1 & 0.0467 & 23.4 & 0.0466 & 23.5 \\
Autopilot~\cite{rzadca2020autopilot} & 0.0829 & 149.5 & 0.0836 & 32.7 & 0.0836 & 32.7 \\
\bottomrule
\end{tabular}
\end{table}

\paragraph{Impact of Phase 1 Substitution.}
Replacing our probabilistic predictor with standard forecasters (TFT, Informer) yields similar MAE but higher Regret (e.g., TFT Regret +175\% on GenAI-Mem). This reflects \textit{prediction--decision misalignment}: forecasting losses under censoring fail to capture tail risk. In contrast, our NLL-based predictor provides calibrated uncertainty ($\sigma_t$) for risk-aware buffering.

\paragraph{Impact of Phase 3 Substitution.}
\textbf{Conformal Prediction} (Phase 3 Subst.; no Phase-2 pretraining) underperforms (Regret +1279\% on GenAI-Mem) because one-sided calibration on \textit{censored residuals} inherits the selection bias in Proposition 1. PIC-RL instead uses IMR to correct the censored feedback signal.

\paragraph{Impact of Phase 2/3 Substitution.}
Substituting the RL module reveals two distinct failure modes.
\textbf{CQL (Phase 2 Subst.)} exhibits ``conservative collapse'' (Figure~\ref{fig:trajectory-comparison}d): without dual-timescale calibration, it converges to a static policy that misses bursts.
\textbf{Pensieve-style Actor--Critic and Thompson Sampling (Phase 2+3 Subst.)} degrade under censoring (e.g., TS Regret +525\% on GenAI-Mem) because generic exploration/control (even with interface-compatible state/action parameterizations) lacks IMR-based counterfactual learning from censored events and can be destabilized by prolonged information blackouts.

\paragraph{Full-Pipeline Analysis.}
Full-stack baselines exhibit structural deficits. \textbf{SPO+} drifts under online censoring because its decision-focused loss assumes fully observed labels. \textbf{Autopilot} over-provisions to reduce shortages, yielding the highest Regret and MAE (Figure~\ref{fig:trajectory-comparison}b).

\begin{figure}[t]
\centering
\includegraphics[width=\linewidth]{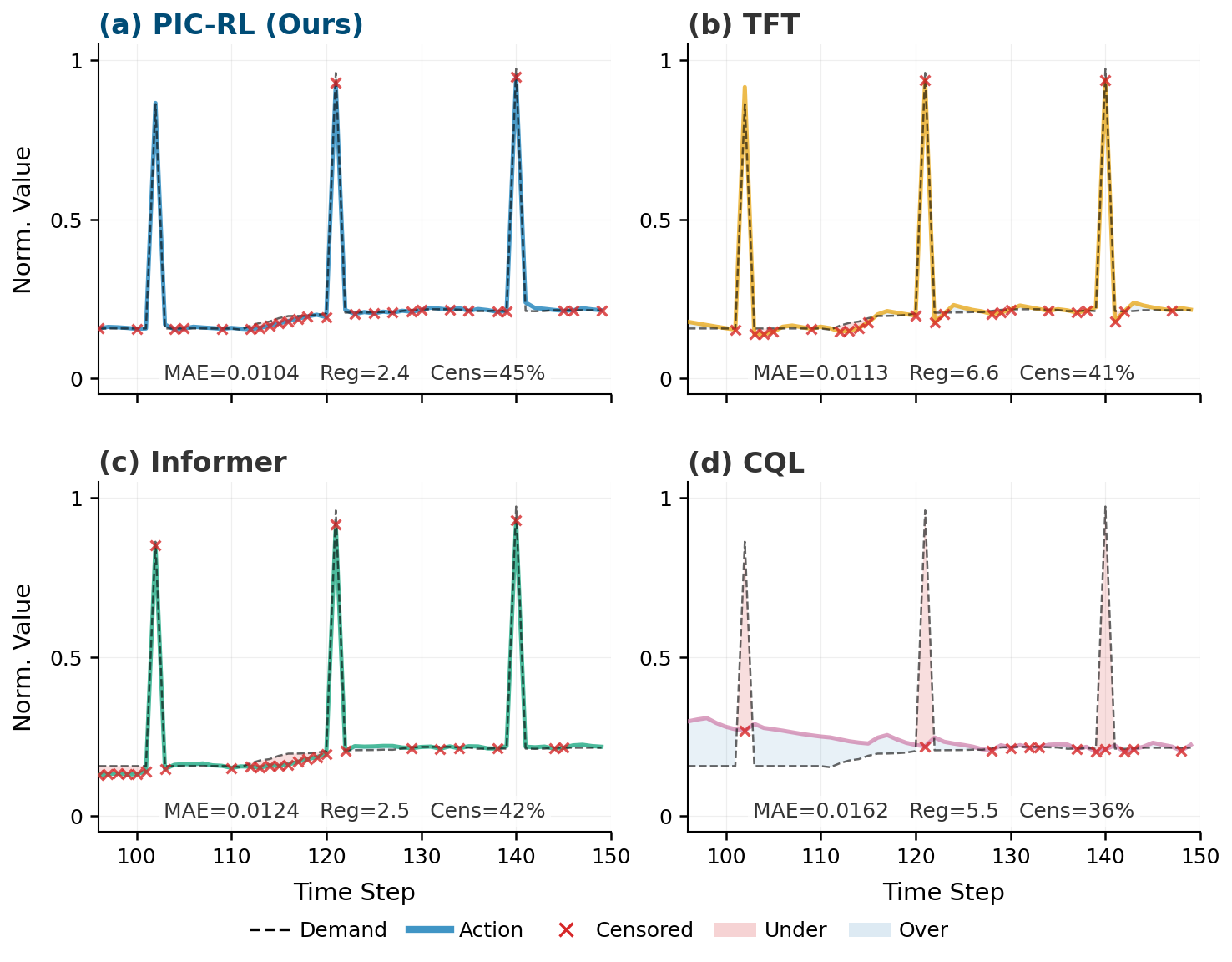}
\caption{Trajectory Comparison (GenAI-Mem). PIC-RL (top) tightly tracks demand. In contrast, CQL (Phase 2 Subst.) exhibits "conservative collapse", while TFT (Phase 1 Subst.) captures periodicity but underestimates stochastic bursts.}
\label{fig:trajectory-comparison}
\end{figure}

\subsection{Ablation Study}
\label{sec:ablation}

An ablation study is conducted to attribute performance gains to individual components. The results (Table~\ref{tab:ablation-full} and Figure~\ref{fig:ablation-bars}) are consistent with the three-phase design: \textit{Foundational} (Phase 1), \textit{Structural} (Phase 2), and \textit{Stabilizing} (Phase 3).

\paragraph{Foundational Role of Uncertainty (A1).}
Removing uncertainty quantification (A1: w/o Uncertainty) yields the largest degradation. As shown in Figure~\ref{fig:ablation-bars}, dropping the probabilistic output $\sigma_t$ increases MAE by \textbf{646\%} on GenAI-Ctx and \textbf{328\%} on GenAI-Mem (relative to Full). Without $\sigma_t$, the controller cannot form an uncertainty-based buffer or compute the IMR-based surrogate term, effectively reducing decisions to a point estimate under asymmetric costs. This supports the necessity of uncertainty quantification for decision-making under endogenous censoring.

\paragraph{Structural Necessity of Pessimism (A2--A3).}
Even with access to uncertainty, the agent requires a mechanism to translate risk into correct gradients.
\begin{itemize}
    \item \textbf{w/o Censored Reward (A2):} Removing the IMR-based surrogate reward increases MAE by \textbf{+57\%} (DLRM), \textbf{+115\%} (GenAI-Ctx), and \textbf{+55\%} (GenAI-Mem), indicating that censored steps provide insufficient directional learning signal without counterfactual surrogate feedback.
    \item \textbf{w/o Pessimism (A3):} Removing the pessimism factor $\Psi(n)$ increases MAE by \textbf{+53\%} (DLRM), \textbf{+113\%} (GenAI-Ctx), and \textbf{+40\%} (GenAI-Mem), indicating that gradient amplification is needed to escape persistent censoring traps.
\end{itemize}

\paragraph{Control-Theoretic Stability (A4--A7).}
The remaining ablations (w/o $k\sigma$, KL, EMA, Pretraining) are stability refinements. Relative to Full, these changes increase MAE by \textbf{9\%--102\%} across datasets (Table~\ref{tab:ablation-full}). For example, on GenAI-Ctx, removing KL (A5) increases MAE by \textbf{102\%}, whereas on GenAI-Mem, removing EMA (A6) increases MAE by \textbf{9\%}. Removing KL regularization (A5) or exponential moving average (EMA) (A6) also increases action variance, consistent with the role of dual-timescale stabilization.

\begin{figure}[t]
\centering
\includegraphics[width=\linewidth]{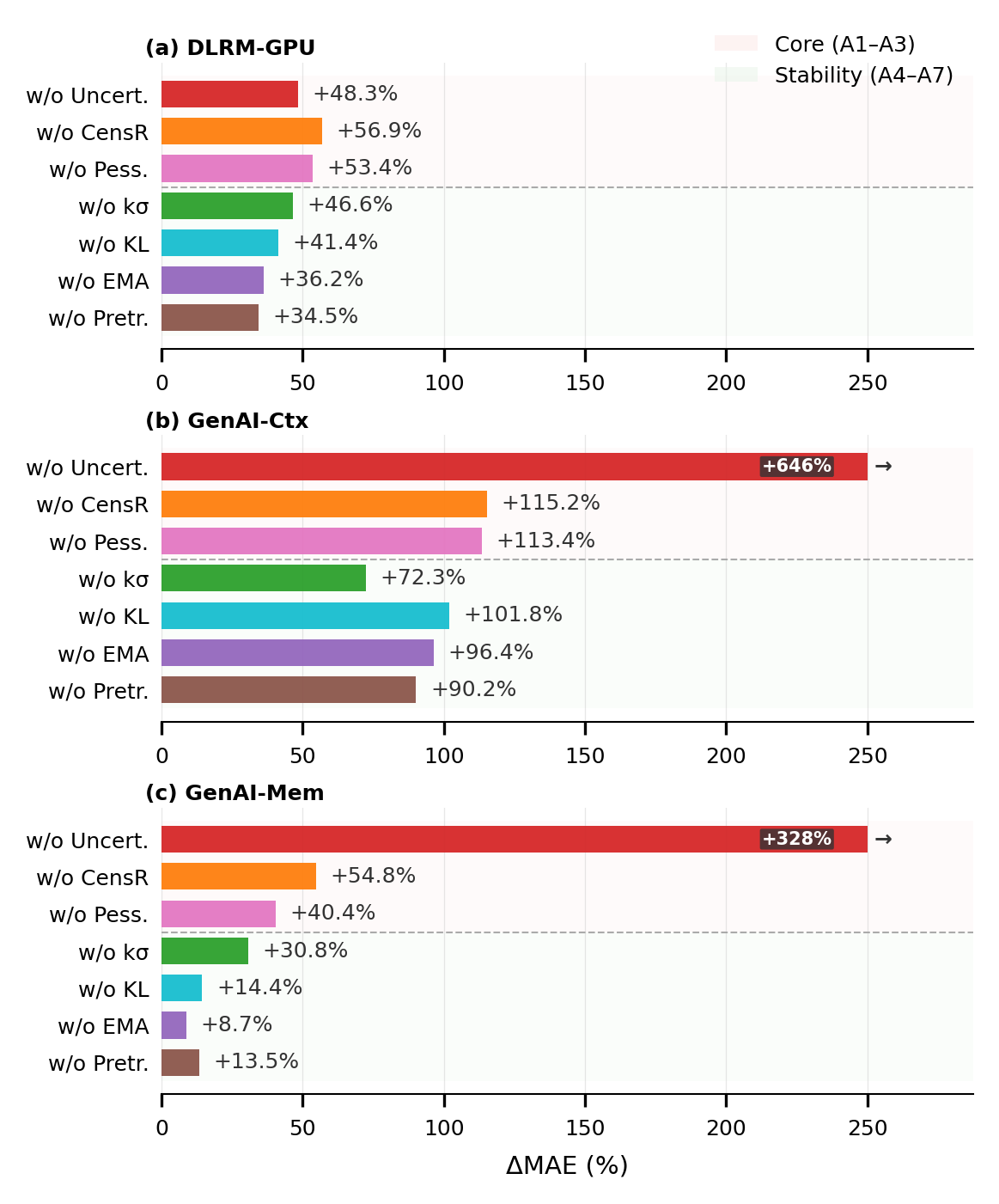}
\caption{Ablation Study (MAE Impact). Uncertainty removal (A1) yields the largest degradation, while removing pessimism (A3) or IMR (A2) weakens adaptation. Stability refinements (A4--A7) provide consistent robustness gains.}
\label{fig:ablation-bars}
\end{figure}

\begin{table}[t]
\centering
\caption{Complete Ablation Results. Metrics degrade consistently across ablations. Uncertainty (A1) and Censored Reward (A2) are the most critical components.}
\label{tab:ablation-full}
\resizebox{\columnwidth}{!}{
\begin{tabular}{l|cc|cc|cc}
\toprule
 & \multicolumn{2}{c|}{\textbf{DLRM}} & \multicolumn{2}{c|}{\textbf{GenAI-Ctx}} & \multicolumn{2}{c}{\textbf{GenAI-Mem}} \\
\textbf{Method} & \textbf{MAE} & \textbf{Regret} & \textbf{MAE} & \textbf{Regret} & \textbf{MAE} & \textbf{Regret} \\
\midrule
\textbf{PIC-RL (Full)} & \textbf{0.0058} & \textbf{22.8} & \textbf{0.0112} & \textbf{5.4} & \textbf{0.0104} & \textbf{2.4} \\
\midrule
\multicolumn{7}{l}{\textit{Foundational \& Structural}} \\
A1: w/o Uncert. & 0.0086 & 25.6 & 0.0836 & 7.3 & 0.0445 & 3.4 \\
A2: w/o CensRwd & 0.0091 & 28.2 & 0.0241 & 7.4 & 0.0161 & 3.8 \\
A3: w/o Pessimism & 0.0089 & 26.3 & 0.0239 & 6.8 & 0.0146 & 3.1 \\
\midrule
\multicolumn{7}{l}{\textit{Control Stability}} \\
A4: w/o $k\cdot\sigma$ & 0.0085 & 24.6 & 0.0193 & 6.1 & 0.0136 & 3.5 \\
A5: w/o KL & 0.0082 & 24.1 & 0.0226 & 5.8 & 0.0119 & 2.7 \\
A6: w/o EMA & 0.0079 & 23.3 & 0.0220 & 5.7 & 0.0113 & 2.6 \\
A7: w/o Pretrain & 0.0078 & 23.5 & 0.0213 & 5.6 & 0.0118 & 2.7 \\
\bottomrule
\end{tabular}
}
\end{table}

\section{Conclusion}
\label{sec:conclusion}

This work formalizes \textbf{Prediction-Induced Censoring}, a pervasive challenge in \textit{online decision support systems} where decisions shape which data become observable. Under PIC, standard learning pipelines can enter a self-reinforcing loop: under-provisioning reduces label availability, amplifies selection bias, and leads to persistent service degradation.

To break this loop, \textbf{PIC-RL} is introduced as a closed-loop decision support framework that transforms censoring from a data quality issue into a usable decision signal. The design couples uncertainty-aware prediction for managing the information--cost trade-off, conservative surrogate feedback for learning directionally correct updates from shortage events (Prop.~1--2), and dual-timescale online adaptation for stable operation under drift. Experiments on production GenAI traces show consistent improvements over strong baselines, reducing service degradation by up to 50\% while maintaining cost efficiency.

\paragraph{Future Directions.}
Future work includes extending the framework to \textbf{multi-resource coupled constraints} (e.g., joint GPU memory and compute censoring) and strengthening \textbf{operational governance} for decision support, such as monitoring, auditability, and safe adaptation under distribution shift.

\bibliographystyle{unsrt}
\bibliography{references}

\end{document}